\documentclass{amsart}

\usepackage{latexsym,amsfonts} 
\usepackage{amsmath,amssymb,graphics,setspace}
\usepackage{amsthm}
\usepackage{MnSymbol}
\usepackage{mathrsfs} 
\usepackage{fontenc}%

\usepackage{paralist}
\usepackage{color}
\usepackage[verbose]{wrapfig}
\usepackage[ruled,vlined,linesnumbered]{algorithm2e}
\usepackage{marvosym}
\usepackage{array}
\usepackage{picins,graphicx}
\usepackage{pstricks,pst-text,pst-grad,pst-node,pst-3dplot,pstricks-add,pst-poly,pst-coil,pst-tree} 
\usepackage{multido,pst-node,pst-bspline}
\usepackage{multirow}
\usepackage{caption}
\usepackage{subfigure}
\usepackage[all]{xy}
\renewcommand{\thesubfigure}{\thefigure.\arabic{subfigure}}
\makeatletter
\renewcommand{\p@subfigure}{}
\renewcommand{\@thesubfigure}{\thesubfigure:\hskip\subfiglabelskip}
\makeatother

\usepackage{hyperref}
\hypersetup{linktocpage=true,colorlinks=true,linkcolor=blue,citecolor=blue,pdfstartview={XYZ 1000 1000 1}}

\newtheorem{example}{Example}

\newtheorem{definition}{Definition}

\newtheorem{theorem}{Theorem}

\usepackage[normalem]{ulem} 

\begin{document}

\title[Descriptive Union]{Descriptive Unions.  A Fibre Bundle Characterization of
the Union of\\ Descriptively Near Sets}

\author[M.Z. Ahmad]{M.Z. Ahmad$^{\alpha}$}
\email{ahmadmz@myumanitoba.ca}
\address{\llap{$^{\alpha}$\,}
Computational Intelligence Laboratory,
University of Manitoba, WPG, MB, R3T 5V6, Canada}
\thanks{\llap{$^{\alpha}$\,}The research has been supported by University of Manitoba Graduate Fellowship and Gorden P. Osler Graduate Scholarship.}

\author[J.F. Peters]{J.F. Peters$^{\beta}$}
\email{James.Peters3@umanitoba.ca}
\address{\llap{$^{\beta}$\,}
Computational Intelligence Laboratory,
University of Manitoba, WPG, MB, R3T 5V6, Canada and
Department of Mathematics, Faculty of Arts and Sciences, Ad\.{i}yaman University, 02040 Ad\.{i}yaman, Turkey}
\thanks{\llap{$^{\beta}$\,}The research has been supported by the Natural Sciences \&
Engineering Research Council of Canada (NSERC) discovery grant 185986 
and Instituto Nazionale di Alta Matematica (INdAM) Francesco Severi, Gruppo Nazionale per le Strutture Algebriche, Geometriche e Loro Applicazioni grant 9 920160 000362, n.prot U 2016/000036.}


\date{}

\dedicatory{P. Alexandroff and S.A. Naimpally}

\begin{abstract}
This paper introduces an extension of descriptive intersection and provides a framework for descriptive unions of nonempty sets.   Fibre bundles provide structures that characterize spatially near as well as descriptively near sets, their descriptive intersection and their unions.   The properties of four different forms of descriptive unions are given.   A main result given in this paper is the equivalence between ordinary set intersection and a descriptive union.   Applications of descriptive unions are given with respect to Jeffs-Novik convex unions and descriptive unions in digital images.    
\end{abstract}

\maketitle
\section{Introduction}\label{sec:intro}
Descriptively near sets were introduced in \cite{Peters2007AMS,Peters2007FINearSets}, elaborated on in~\cite{DiConcilio2018MCSdescriptiveProximities} and applied in a number of different contexts such as shape classification~\cite{Peters2018AlMSproximalPlanarShapes}~\cite{AhmadPeters2017TAMCSspokes},  vortex nerve structures~\cite{Peters2018JMSMvortexNerves}, strongly proximal Edelsbrunner-Harer nerves~\cite{Peters2016MNC}, quantum entanglement~\cite{PetersTozzi2916quantumEntanglement} and in the foundations of computational proximity~\cite{Peters2016ComputationalProximity}.  The introduction of such sets was motivated by the need to provide a framework for representing both spatially as well as descriptively similar objects. A \emph{description} of an object is a real-valued feature vector that characterizes the object. Digital images are a fertile ground for both spatial and descriptive nearness. For example adjacent pixels are spatially near and pixels with the same intensity values are descriptively near regardless of their spatial location. 
	
	To formulate spatial and descriptive nearness of sets, we require a structure that links a set with its description. A \emph{probe function} is a map $\phi:2^K \rightarrow \mathbb{R}^n$, which assigns a $n-$dimensional vector valued description
	to all subsets of $K$. A \emph{fibre bundle},$(K_\Phi,K,\pi,\phi(U))$, can be used to specify the link between a set and its description. Here $K_\Phi$ is the \emph{glossa} (Latin for glossary) , $K$ is the set, $\pi$ is a continuous surjection and $\phi(U), U\subset K$ is the \emph{fibre}. $K_\phi$ is a set in which each element $k$ of $K$ is paired with its description $\phi(k)$(the \emph{fibre}), hence we call this structure a \emph{glossa}.  This structure can be represented as the following diagram.
	\begin{align}\label{eq:fibre}
	\xymatrix{ \phi(U) \ar[r] & K_\Phi \ar[r]^\pi & K }
	\end{align}
	Moreover, this structure satisfies the \emph{local trivialization} property. Hence, the following diagram commutes.
	\begin{align}\label{eq:localtriv}
	\xymatrix{ K_\Phi \supset \pi^{-1}(U) \ar[r]^\gamma \ar[d]_\pi& U \times \phi(U) \ar[dl]^{{proj}_1}\\
	K \supset U}
	\end{align}
	This means that in a small neighborhood $U \subset K$, $K_\Phi$ is homeomorphic( under map $\gamma$) to $U \times \phi(U)$. Note that $\pi: K_\Phi \rightarrow K$ maps each element in the glossa $K_\Phi$ to the corresponding element in $K$, in this small neighborhood $\pi^{-1}(U) \subset K$. This is due to the fact that a fibre bundle is locally homomorphic to a product space but can vary in structure globally. We also assume that every element of a set has a description i.e. $\phi(k) \neq \emptyset$ for $k \in K$.
\section{Descriptive set theoretic operations}	
In this work, it is important to illustrate the distinction between $K,K_\Phi$ and $\phi(K)$. This distinction is centeral to the proposed framework that unifies the spatial and descriptive aspects of a set. To ellaborate this distinction we present the following example.
\begin{example}
	Suppose we have four balls of three different colors, $B_1,B_2, B_3$ and $B_4$. We can represent these objects as a set $K=\{B_1,B_2,B_3,B_4\}$. Now suppose we have a function,$\phi$ that measures the color of each of these balls. Now colors of the balls can be written as a set $\phi(K)=\{\phi(k): k \in K\}$. In this case $\phi(K)=\{Blue, Black, Blue, Green\}$. The glossa $K_\phi$ is the set in which each of the elements is paired with its description,$K_\Phi=\{\{k,\phi(k)\}:k \in K\}$. In this case $K_\phi=\{\{B_1,Blue\},\{B_2,Black\},\{B_3,Blue\},\{B_4,Green\}\}$. Thus, the objects exist in a broader universe $K_\Phi$, which pairs their spatial aspects $K$ with their respective descriptions $\phi(K)$. \qquad \SquareSteel
\end{example}
 Union and intersection are two important set theoretic operations. Union is defined as
\begin{definition}\label{def:union}
	Let $A,B$ be two sets. Then, 
	\begin{align*}
	A \bigcup B= \{x: x \in A\, or\, x \in B\},
	\end{align*}
	and $\bigcup$ is the union operator.
\end{definition}
Intersection is defined as
\begin{definition}\label{def:intersection}
	Let $A,B$ be two sets. Then, 
	\begin{align*}
	A \bigcap B= \{x: x \in A \, and \, x \in B\},
	\end{align*}
	and $\bigcap$ is the intersection operator.
\end{definition}
Here we will consider the extension of these operations to the framework of sets with descriptions, termed as glossa.
\subsection{Descriptive Intersection}\label{sec:desint}
Notion of descriptive intersection was introduced to study the similarity of sets in terms of their description \cite{peters2013mcs}. It is defined as follows,
\begin{definition}\label{def:desint}
	Let $A,B \subset K$ be two subsets of $K$ and $\phi:2^K \rightarrow \mathbb{R}^n$ be a probe function. Then
	\begin{align*}
	A \mathop \bigcap \limits_\Phi B = \{x\in A \cup B: \phi(x)\in \phi(A)\,and \, \phi(x) \in \phi(B) \}
	\end{align*}
	where $\mathop \bigcap \limits_\Phi$ is the descriptive intersection.
\end{definition}
We can represent this notion in the following diagram,
\begin{align}\label{eq:desint}
\xymatrix{
    A \ar[d]^\phi & A \cup B \ar[l]_a \ar[r]^b &B \ar[d]^\phi \\
    \phi(A) \ar@<1ex>[u]^\pi \ar[r]^c& \phi(A) \cap \phi(B) \ar[d]^\pi& \phi(B) \ar@<1ex>[u]^\pi \ar[l]_d \\
    & A \mathop \bigcap \limits_\Phi B & 
}
\end{align}
Let us study some important properties of $\mathop \bigcap \limits_\Phi$.
\begin{theorem}\label{thm:desint_prop}
Let $A,B \subset K$ be two subsets of a set $K$,$\phi:2^K \rightarrow \mathbb{R}^n$ be the probe function and $\pi: \mathbb{R}^n \rightarrow 2^K$ be a map such that $\pi: x \mapsto \{y \in K: \phi(y)=x\}$. Then, $A \mathop \bigcap \limits_\Phi B$ has following properties:
\begin{compactenum}[$1^0$]
	\item $A \mathop \bigcap \limits_\Phi B = A \mathop \bigcap \limits_\Phi B$.
	\item $A=\emptyset \Rightarrow A \mathop \bigcap \limits_\Phi B= \{x \in B:\phi(x) = \phi(\emptyset) \}$.
	\item $A=B \Rightarrow A \mathop \bigcap \limits_\Phi B = A$.
	\item $A \cap B \Rightarrow A \mathop \bigcap \limits_\Phi B$.
	\item $A \mathop \bigcap \limits_\Phi B \not\Rightarrow A \cap B$.
	\item $(A \mathop \bigcap \limits_\Phi B = A \cap B) \Leftrightarrow \phi$ is an injective function.
	\item $A \mathop \bigcap \limits_\Phi B \subseteq A \cup B$.
	
\end{compactenum}	
\end{theorem}
\begin{proof}
	\begin{compactenum}[$1^0$]
		\item By interchanging $A$ and $B$ in the arrow diagram illustrated in eq.~\ref{eq:desint}, we get the following diagram:
		\begin{align*}
		\xymatrix{
			B \ar[d]^\phi & B \cup A \ar[l]_b \ar[r]^a &A \ar[d]^\phi \\
			\phi(B) \ar@<1ex>[u]^\pi \ar[r]^d& \phi(B) \cap \phi(A) \ar[d]^\pi& \phi(A) \ar@<1ex>[u]^\pi \ar[l]_c \\
			& B \mathop \bigcap \limits_\Phi A & 
		}
		\end{align*}
		By the def.~\ref{def:union} $A \cup B$ is the set of elements in either $A$ or $B$, and $B \cup A$ is the set of elements in either $B$ or $A$. It is clear that $A \cup B = B \cup A$. Moreover, def.~\ref{def:intersection} $\phi(A) \cap \phi(B)$ is the set of elements common to both $\phi(A)$ and $\phi(B)$. $\phi(B) \cap \phi(A)$ is the set of elements common to both $\phi(B)$ and $\phi(A)$. Hence, $\phi(A) \cap \phi(B) = \phi(B) \cap \phi(A)$. Now, as $\pi(\phi(A) \cap \phi(B)= A \mathop \bigcap \limits_\Phi B$ and $\phi(A)\cap \phi(B)= \phi(B) \cap \phi(A)$, this implies that $\pi(\phi(A)\cap \phi(B))=\pi(\phi(B) \cap \phi(A))$. It is known that $\pi(\phi(B)\cap \phi(A))=B \mathop \bigcap \limits_\Phi A$, thus $A \mathop \bigcap \limits_\Phi B = B \mathop \bigcap \limits_\Phi A$.
		\item By subsituting $A=\emptyset$ in the diagram in eq.~\ref{eq:desint}, we get the following:
		\begin{align*}
		\xymatrix{
			\emptyset \ar[d]^\phi & \emptyset \cup B \ar[l]_a \ar[r]^b &B \ar[d]^\phi \\
			\phi(\emptyset) \ar@<1ex>[u]^\pi \ar[r]^c& \phi(\emptyset) \cap \phi(B) \ar[d]^\pi& \phi(B) \ar@<1ex>[u]^\pi \ar[l]_d \\
			& A \mathop \bigcap \limits_\Phi B & 
		}
		\end{align*}
		The probe function is defined as $\phi:2^K \rightarrow \mathbb{R}^n$. Thus, $\phi(\emptyset)$ is defined even when $\emptyset \not\in B$. $\phi(\emptyset) \cap \phi(B) \neq \emptyset$, as $\emptyset \subset 2^B$ thus, $\phi(\emptyset) \in \phi(B) 
		$. This is due to the fact that $\emptyset$ is by definition included in the power set $2^K$ of every set $K$. It can then be established that $\pi(\phi(\emptyset) \cap \phi(B))=\pi(\phi(\emptyset))=\{x \in B : \phi(x)=\phi(\emptyset)\}$, as per the definition of $\pi$.
	    \item By subsitutting $A=B$ in the arrow diagram of eq.~\ref{eq:desint} we get,
	    \begin{align*}
	    \xymatrix{
	    	A \ar[d]^\phi & A \cup A \ar[l]_a \ar[r]^b &A \ar[d]^\phi \\
	    	\phi(A) \ar@<1ex>[u]^\pi \ar[r]^c& \phi(A) \cap \phi(A) \ar[d]^\pi& \phi(A) \ar@<1ex>[u]^\pi \ar[l]_d \\
	    	& A \mathop \bigcap \limits_\Phi A & 
	    }
	    \end{align*}
		We know from def.~\ref{def:intersection} that $A \cap A =A$, thus $\phi(A) \cap \phi(A)=\phi(A)$. Moreover, $\pi(\phi(A))=A$ as for each $z \in \phi(A)$, $\pi: z \mapsto \{x \in A: \phi(x) = z\}$. 
		\item As we know that $\phi$ is a function, thus by definition it can be one-to-one or many-to-one but not one-to-many. Moreover, $A \cap B$ means that there are elements common to both the sets. Thus due to the definition of a funciton $\phi(A \cap B) \subset \phi(A)$ and $\phi(A \cap B) \subset \phi(B)$ leading to $\phi(A \cap B) \subset \phi(A) \cap \phi(B)$. Hence, $\phi(A) \cap \phi(B) \neq \emptyset$. Thus, $A \cap B$ implies $A \mathop \bigcap \limits _\Phi B$.
		
	 \item As we know that $\phi$ being a function can be either one-to-one or many-to-one but not one to many. Keeping this in mind, and the fact that $A \mathop \bigcap \limits_\Phi B =\pi(\phi(A) \cap \phi(B))$. We can see that $\phi$ is a many-to-one function then it is possible that for $A,B \subset K$, $\phi(a \in A)=\phi(b \in B)$ without $a=b$. Which allows for $a,b \not\in A \cap B$. Thus, $A \mathop \bigcap \limits_\Phi B \not \Rightarrow A \cap B$.
	 
	 \item 
	 \begin{compactenum}
	 	\item[$\Rightarrow:$] We know that $A \mathop \bigcup \limits_{\Phi} B=\pi(\phi(A) \cap \phi(B))$. Thus, for $A \mathop \bigcap \limits_\Phi B=A \cap B$, each element  of $A \cup B$ must have a unique description. Hence, $\phi$ must be a one-to-one or an injective function.
	 	\item[$\Leftarrow:$] $\phi$ being an injection means that each element of a set has a unique description. Thus, only way for an element in $\phi(A)$ to be the same as an element in $\phi(B)$, is to be in both $A$ and $B$. Hence, $A \mathop \bigcap \limits_\Phi B= A \cap B$. 
	 	\end{compactenum}
 	Proving both $(A \mathop \bigcap \limits_\Phi B= A \cap B) \Rightarrow \phi \text{ is an injection}$ and $\phi \text{ is an in injection} \Rightarrow (A \mathop \bigcap \limits_\Phi B= A \cap B)$, leads to the conclusion that bijection,$\Leftrightarrow$, is true.
 	
 	\item This is fact obvious from the def.~\ref{def:desint}. All the elements considered for the membership of $A \mathop \bigcap \limits_\Phi B$ are elements of $A \cup B$. Then, we further narrow our search by choosing the elements such that $\phi(x) \in \phi(A)$ and $\phi(x) \in \phi(B)$.
	\end{compactenum}

\end{proof}
From this we can see that $A \mathop \bigcap \limits_\Phi B$ is generally distinct from both $A \cap B$ and $A \cup B$. Thus, it yields information that is different from both the spatial union and intersection.
\subsection{Descriptive Union}\label{sec:desunion}
Now, let us move onto the idea of a descriptive union. We will define different notions of union based on whether they are restrictive or non-restrictive(spatially), and descriptive discriminatory or nondiscriminatory. First, we define what these terms stand for as follows.
\begin{itemize}
	\item \textbf{restrictive:} all the elements in $A \cup B$ are considered.
	\item \textbf{non-restrictive:} only the elements in $A \cap B$ are considered.
	\item \textbf{descriptive nondiscriminatory:} we consider element with any value of description.
	\item \textbf{descriptive discriminatory:} we consider elements with specific values of description.
\end{itemize}
For sets $A,B \subset K$, the four different types of unions yielded by the above mentioned categorization are listed in table~\ref{table:desun}.
\begin{table}[h!]
	\centering
	\begin{tabular}{|m{1em}m{1em}|c|c|}
		\cline{1-4}
		& & \multicolumn{2}{c|}{\small Spatially} \\ \cline{3-4}
		& & \small restrictive & \small non-restrictive  \\ \cline{1-4}
		\multirow{2}{*}{\rotatebox[origin=c]{90}{\parbox[t][\height][c]{1.8cm}{ \small Descriptive}}}& \multicolumn{1}{|c|}{\rotatebox[origin=c]{90}{\parbox[t]{2.2cm}{\centering \small discriminatory}}} & $A \tilde{\mathop \bigcup \limits_{\phi=\{i,j\}}} B$ & $A \mathop \bigcup \limits_{\phi=\{i,j\}} B$\\
		\cline{2-4}
		& \multicolumn{1}{|c|}{\rotatebox[origin=c]{90}{\parbox[t]{2.6cm}{\centering \small nondiscriminatory}}}& $A \tilde{\mathop \bigcup \limits_\Phi} B$ & $A \mathop \bigcup \limits_\Phi B$  \\
		\hline
	\end{tabular}
\caption{Categorizing different types of descriptive unions}
\label{table:desun}
\end{table}
Now let us explore each type of descriptive union in detail.
\subsubsection{Restrictive and descriptive discriminatory union}\label{sec:def1}
~\\
\begin{definition}\label{def:desun1}
	Let $A \subset K$ be a subset in $K$ and $\phi:2^K \rightarrow \mathbb{R}^n$ be a probe function. Then
	\begin{align*}
	A \tilde{\mathop \bigcup \limits_{\phi=\{i,j\}}} B = \{x \in A \cap B: \phi(x)=i \, or \,\phi(x)=j \}, 
	\end{align*}
	where $\tilde{\mathop \bigcup \limits_{\phi=\{i,j\}}}$ is the spatially restricted and descriptively discriminant union.
\end{definition}
This definition leads to the following diagram.
\begin{align}\label{eq:desun1}
\centering
\xymatrix{
	A \ar[dd]^\phi & & A \cap B \ar[ll]_a \ar[rr]^b & &B \ar[dd]^\phi \\
	& \{i\} \ar[rd]_g & &\{i\} \ar[ld]_k & \\
	\phi(A) \ar@<1ex>[uu]^\pi \ar[ru]^c \ar[rd]^d& & \{i\} \cup \{j\} \ar[dd]^\pi& & \phi(B) \ar@<1ex>[uu]^\pi \ar[lu]_e \ar[ld]_f \\
	& \{j\} \ar[ru]_h &  & \{j\} \ar[lu]_l & \\
	& & A \tilde{\mathop \bigcup \limits_{\phi=\{i,j\}}} B & & 
}
\end{align}
Out of the elements of $A \cap B$, we only consider those with descriptions equal to $i$ or $j$. These values are selected beforehand for this purpose. The following theorem presents important results regarding restrictive and descriptive discriminatory union.
\begin{theorem}
	Let $A,B \subset K$ be subsets in $K$, $\phi:2^K \rightarrow \mathbb{R}^n$ be a probe function and $\pi: \mathbb{R}^n \rightarrow 2^K$ be a map such that $\pi:x \mapsto \{y \in K:\phi(y)=x\}$.Then the following properties are true for the restrictive and descriptive discriminatory union:
	\begin{compactenum}[$1^o$]
		\item $A \tilde{\mathop \bigcup \limits_{\phi=\{i,j\}}} B = B \tilde{\mathop \bigcup \limits_{\phi=\{i,j\}}} A$.
		\item $A=\emptyset \Rightarrow A \tilde{\mathop \bigcup \limits_{\phi=\{i,j\}}} B= \{ \emptyset: \phi(\emptyset)=i \; or \; \phi(\emptyset)=j \}$. 
		\item $A=B \Rightarrow A \tilde{\mathop \bigcup \limits_{\phi=\{i,j\}}} B=\{x \in A: \phi(x)=i \; or \; \phi(x)=j \}$.
		\item $\textbf{Codomain}(\phi|_{A \cap B})=\{i,j\}\Leftrightarrow (A \tilde{\mathop \bigcup \limits_{\phi=\{i,j\}}} B = A \cap B)$.
		\item  $\{i,j\} \not\in \textbf{Codomain}(\phi|_{A \cap B}) \Leftrightarrow A \tilde{\mathop \bigcup \limits_{\phi=\{i,j\}}} B = \emptyset$.
		\item $i \not\in \textbf{Codomain}(\phi|_{A \cap B}) \Leftrightarrow (A \tilde{\mathop \bigcup \limits_{\phi=\{i,j\}}} B = A \tilde{\mathop \bigcup \limits_{\phi=\{j\}}} B)$.
	\end{compactenum}
	where $\phi|_{A \cap B}$ is the restriction of $\phi$ to the domain $A \cap B$.
\end{theorem}
\begin{proof}
	\begin{compactenum}[$1^o$]
		\item By subsituting $A$ with $B$ and vice versa in eq.~\ref{eq:desun1} we get the following diagram:
			\begin{align*}
			\centering
			\xymatrix{
				B \ar[dd]^\phi & & B \cap A \ar[ll]_a \ar[rr]^b & &A \ar[dd]^\phi \\
				& \{i\} \ar[rd]_k & &\{i\} \ar[ld]_g & \\
				\phi(B) \ar@<1ex>[uu]^\pi \ar[ru]^e \ar[rd]^f& & \{i\} \cup \{j\} \ar[dd]^\pi& & \phi(A) \ar@<1ex>[uu]^\pi \ar[lu]_c \ar[ld]_d \\
				& \{j\} \ar[ru]_l &  & \{j\} \ar[lu]_h & \\
				& & B \tilde{\mathop \bigcup \limits_{\phi=\{i,j\}}} A & & 
			}
			\end{align*}
		 We know that elements in $B$ and $A$ are the same as the elements in $A$ and $B$, hence $A \cap B=B \cap A$. The rest of the diagram is identical to eq.~\ref{eq:desun1}. From this we can conclude that $A \tilde{\mathop \bigcup \limits_{\phi=\{i,j\}}} B$ equals $\pi(\{i\}\cup \{j\})$ which equals $B \tilde{\mathop \bigcup \limits_{\phi=\{i,j\}}} A$. Hence $A \tilde{\mathop \bigcup \limits_{\phi=\{i,j\}}} B=B \tilde{\mathop \bigcup \limits_{\phi=\{i,j\}}} A$.	
		\item Let us substitute $A= \emptyset$ in eq.~\ref{eq:desun_2}, to get the following diagram:
			\begin{align*}
			\centering
			\xymatrix{
				\emptyset \ar[dd]^\phi & & \emptyset \cap B \ar[ll]_a \ar[rr]^b & &B \ar[dd]^\phi \\
				& \{i\} \ar[rd]_g & &\{i\} \ar[ld]_k & \\
				\phi(\emptyset) \ar@<1ex>[uu]^\pi \ar[ru]^c \ar[rd]^d& & \{i\} \cup \{j\} \ar[dd]^\pi& & \phi(B) \ar@<1ex>[uu]^\pi \ar[lu]_e \ar[ld]_f \\
				& \{j\} \ar[ru]_h &  & \{j\} \ar[lu]_l & \\
				& & \emptyset \tilde{\mathop \bigcup \limits_{\phi=\{i,j\}}} B & & 
			}
			\end{align*}
			We can observe from this diagram that $\emptyset \tilde{\mathop \bigcup \limits_{\phi=\{i,j\}}} A $ is embedded in $\emptyset \cap A=\emptyset$. But, as we know that $\phi:2^K \rightarrow \mathbb{R}^n$, thus we have a description for $\emptyset$ even if it is not an element of the sets in question. This is because $\emptyset \in 2^K$, as empty set is a subset of every set. From the definition of $\pi(\{i\}\cup\{j\})$ will equal $\{\emptyset\}$ if $\phi(\emptyset)=i$ or $\phi(\emptyset)=j$. We are using the notion that $X=\emptyset$ and $X=\{\emptyset\}$ are different with the former being an empty set while the later having one element which is the empty set itself. 
		\item By substituting $B=A$ in eq.~\ref{eq:desun_2} we get the following diagram:
			\begin{align*}
			\centering
			\xymatrix{
				A \ar[dd]^\phi & & A \cap A \ar[ll]_a \ar[rr]^b & &A \ar[dd]^\phi \\
				& \{i\} \ar[rd]_g & &\{i\} \ar[ld]_k & \\
				\phi(A) \ar@<1ex>[uu]^\pi \ar[ru]^c \ar[rd]^d& & \{i\} \cup \{j\} \ar[dd]^\pi& & \phi(A) \ar@<1ex>[uu]^\pi \ar[lu]_e \ar[ld]_f \\
				& \{j\} \ar[ru]_h &  & \{j\} \ar[lu]_l & \\
				& & A \tilde{\mathop \bigcup \limits_{\phi=\{i,j\}}} A & & 
			}
			\end{align*}
			It can be observed that $A \tilde{\mathop \bigcup \limits_{\phi=\{i,j\}}} A$ is embedded in $A \cap A=A$. Using the definition of $\pi$, and def.~\ref{def:desun1} we can see that we get the elements of $A \cap A=A$ that have a description of $i$ or $j$.
		\item \textbf{Codomain}$(\phi|_{A \cap B})=\{i,j\}$ means that all the elements in $A \cap B$ have a description of either $i$ or $j$. This reduced def.~\ref{def:desun1} to $A \cap B$ as every element in this region satisfies the conditions on description.
		\item $\{i,j\} \not\in \textbf{Codomain}(\phi|_{A\cap B})$ means that no element in $A \cap B$ have a description of $i$ or $j$. Thus, def.~\ref{def:desun1} dictates that $A \tilde{\mathop \bigcup \limits_{\phi=\{i,j\}}} B= \emptyset$.
		\item 
			\begin{compactenum}
				\item[$\Rightarrow:$] $i \not\in \textbf( Codomain)(\phi|_{A \cap B})$ means that no element in $A \cap B$ has the deescription $i$. Under this assumption we can see that def.~\ref{def:desun1} will reduce to the definition of $A \tilde{\mathop \bigcup \limits_{\phi=\{j\}}} B$. Because all the elements having a description of $i$ or $j$ are the same as all the elements having the description $j$. This is due to the fact that no element has the description $i$.
				
				\item[$\Leftarrow:$] $A \tilde{\mathop \bigcup \limits_{\phi=\{i,j\}}} B=A \tilde{\mathop \bigcup \limits_{\phi=\{j\}}} B$, means that the elements in $A \cap B$ having a description $i$ or $j$ are the same as the elements having description $j$. This can only happen if none of the elements in $A \cap B$ have the description $i$. Hence, $i \in \textbf{ Codomain}(\phi|_{A \cap B})$.				
			\end{compactenum}
		proving both $\Rightarrow$ and $\Leftarrow$ leads to the conclusion that $\Leftrightarrow$ holds.
				
	\end{compactenum}
\end{proof}
\subsubsection{Non-restrictive and descriptive discriminatory union}\label{sec:def2}
~\\
Consider taking into account, all the unique elements of sets $A$ and $B$. Then, select those elements which have either of the two values of description selected apriori. This particular notion of a union is termed, spatially inrestricted as it starts from $A \cup B$ and descriptively discriminant, as it then selects elements having either of the two descriptions chosen beforehand. The resulting definition is   
\begin{definition}\label{def:desun2}
	Let $A,B \subset K$ be subsets in $K$ and $\phi:2^K \rightarrow \mathbb{R}^n$ be a probe function. Then
	\begin{align*}
	A \mathop \bigcup \limits_{\phi=\{i,j\}} B = \{x \in A \cup B: \phi(x)=i \, or \,\phi(x)=j\}, 
	\end{align*}
	where $\mathop \bigcup \limits_{\phi=\{i,j\}}$ is the spatially unrestricted and descriptively discriminant union.
\end{definition}
From the definition we can construct the following .
\begin{align}\label{eq:desun_2}
\centering
\xymatrix{
	A \ar[dd]^\phi & & A \cup B \ar[ll]_a \ar[rr]^b & &B \ar[dd]^\phi \\
	& \{i\} \ar[rd]_g & &\{i\} \ar[ld]_k & \\
	\phi(A) \ar@<1ex>[uu]^\pi \ar[ru]^c \ar[rd]^d& & \{i\} \cup \{j\} \ar[dd]^\pi& & \phi(B) \ar@<1ex>[uu]^\pi \ar[lu]_e \ar[ld]_f \\
	& \{j\}\ar[ru]_h &  &\{j\} \ar[lu]_l & \\
	& & A \mathop \bigcup \limits_{\phi=\{i,j\}} B & & 
}
\end{align}
Where $\{i\},\{j\}$ are values of description, chosen before hand for analysis. We only select the elements that have either of these descriptions. 

We formulate important results regarding non-restrictive and descriptive discriminatory union in the following theorem. 
\begin{theorem}
	Let $A,B \subset K$ be subsets in $K$, $\phi:2^K \rightarrow \mathbb{R}^n$ be a probe function and $\pi: \mathbb{R}^n \rightarrow 2^K$ be a map such that $\pi:x \mapsto \{y \in K:\phi(y)=x\}$.Then the following properties are true for the non-restrictive and descriptive discriminatory union:
	\begin{compactenum}[$1^o$]
		\item $A \mathop \bigcup \limits_{\phi=\{i,j\}} B = B \mathop \bigcup \limits_{\phi=\{i,j\}} A$.
		\item $A=\emptyset \Rightarrow A \mathop \bigcup \limits_{\phi=\{i,j\}} B= \{x \in B: \phi(x)=i \, or \, \phi(x)=j \}$. 
		\item $A=B \Rightarrow A \mathop \bigcup \limits_{\phi=\{i,j\}} B=\{x \in A: \phi(x)=i \, or \, \phi(x)=j \}$.
		\item $\textbf{Codomain}(\phi|_{A \cup B})=\{i,j\}\Leftrightarrow (A \mathop \bigcup \limits_{\phi=\{i,j\}} B = A \cup B)$.
		\item $\forall A,B \text{ nonempty subsets of K }\,and\, \{i,j\} \not\in \textbf{Codomain}(\phi|_{A \cup B}) \Leftrightarrow A \mathop \bigcup \limits_{\phi=\{i,j\}} B = \emptyset$.
		\item $i \not\in \textbf{Codomain}(\phi|_{A \cup B}) \Leftrightarrow (A \mathop \bigcup \limits_{\phi=\{i,j\}} B = A \mathop \bigcup \limits_{\phi=\{j\}} B)$.
		\item $\phi \text{ is an injective function and } \phi|_{A \cap B}=\{i,j\} \Leftrightarrow (A \mathop \bigcup \limits_{\phi=\{i,j\}} B = A \cap B)$.
	\end{compactenum}
where $\phi|_{A \cap B}$ is the restriction of $\phi$ to the domain $A \cap B$.
\end{theorem}
\begin{proof}
	\begin{compactenum}[$1^o$]
		\item Interchanging $A$ and $B$ in the diagram illustrated in eq.~\ref{eq:desun_2}, we get the following
			\begin{align*}
			\centering
			\xymatrix{
				B \ar[dd]^\phi & & B \cup A \ar[ll]_b \ar[rr]^a & &A \ar[dd]^\phi \\
				& \{i\} \ar[rd]_k & &\{i\} \ar[ld]_g & \\
				\phi(B) \ar@<1ex>[uu]^\pi \ar[ru]^e \ar[rd]^f& & \{i\} \cup \{j\}  \ar[dd]^\pi& & \phi(A) \ar@<1ex>[uu]^\pi \ar[lu]_c \ar[ld]_d \\
				& \{j\} \ar[ru]_l &  & \{j\} \ar[lu]_h & \\
				& & B \mathop \bigcup \limits_{\phi=\{i,j\}} A & & 
			}
			\end{align*}
			Comparing this to eq.~\ref{eq:desun_2}, we observe that $B \cup A$ is the same as $A \cup B$. This is because $x \in A$ or $x \in B$ is equivalent to $x \in B$ or $x \in A$. The rest of the diagram is same as that in diagram~\ref{eq:desun_2}. Thus, $A \mathop \bigcup \limits_{\phi=\{i,j\}} B= \pi(\{i\}\cup \{j\})$ is equal to $\pi(\{i\} \cup \{j\})= B \mathop \bigcup \limits_{\phi=\{i,j\}} A$.
		
		\item By setting $A=\emptyset$ in the diagram \ref{eq:desun_2} we get,
			\begin{align*}
			\centering
			\xymatrix{
				\emptyset \ar[dd]^\phi & & \emptyset \cup B \ar[ll]_a \ar[rr]^b & &B \ar[dd]^\phi \\
				& \{i\} \ar[rd]_g & &\{i\} \ar[ld]_k & \\
				\phi(\emptyset) \ar@<1ex>[uu]^\pi \ar[ru]^c \ar[rd]^d& & \{i\}\cup \{j\} \ar[dd]^\pi& & \phi(B) \ar@<1ex>[uu]^\pi \ar[lu]_e \ar[ld]_f \\
				& \{j\} \ar[ru]_h &  & \{j\} \ar[lu]_l & \\
				& & \emptyset \mathop \bigcup \limits_{\phi=\{i,j\}} B & & 
			}
			\end{align*}
			As the diagram above shows that $\emptyset \mathop \bigcup \limits_{\phi=\{i,j\}} B$ is embedded in $\emptyset \cup B$. We know that $\emptyset \cup B=B$ and by the definition of $\pi$, we can establish that $ \emptyset \mathop \bigcup \limits_{\phi=\{i,j\}} B$ is equal to $\{x \in B: \phi(x)=i \, or \, \phi(x)=j\}$.
		
		\item By setting $A=B$ in diagram \ref{eq:desun_2}, we get 
		\begin{align*}
		\centering
		\xymatrix{
			A \ar[dd]^\phi & & A \cup A \ar[ll]_a \ar[rr]^b & &B \ar[dd]^\phi \\
			& \{i\} \ar[rd]_g & &\{i\} \ar[ld]_g & \\
			\phi(A) \ar@<1ex>[uu]^\pi \ar[ru]^c \ar[rd]^d& & \{i\}\cup \{j\} \ar[dd]^\pi& & \phi(A) \ar@<1ex>[uu]^\pi \ar[lu]_c \ar[ld]_d \\
			& \{j\} \ar[ru]_h &  & \{j\} \ar[lu]_h & \\
			& & A \mathop \bigcup \limits_{\phi=\{i,j\}} A& & 
		}
		\end{align*}
		This diagram illustrates that $A \mathop \bigcup \limits_{\phi=\{i,j\}} A $ is embedded in $A \cup A=A$. By the definition of $\pi$, we can establish that $A \mathop \bigcup \limits_{\phi=\{i,j\}} A$ is $\{x\in A: \phi(x)=i \, or \, \phi(x)=j\}$.
		
	\item 
		\begin{compactenum}
			\item[$\Rightarrow:$] \textbf{Codomain}$(\phi|_{A \cup B})=\{i,j\}$ implies that all the elements in $A,B$, have a description of either $i$ or $j$. Thus, the def.~\ref{def:desun2} is reduced to $A \cup B$, since every element in $A \cup B$ has a description $i$ or $j$.
			\item[$\Leftarrow:$] Def.~\ref{def:desun2} implies that if $A \mathop \bigcup \limits_{\phi=\{i,j\}} B=A \cup B$, then the description of every element is either $i$ or $j$. Which means that \textbf{Codomain}$(\phi|_{A \cup B})\{i,j\}$. 
		\end{compactenum}
	Proving both $\Rightarrow$ and $\Leftarrow$ portions of the given statement implies that the bijection, $\Leftrightarrow$, also holds.
	\item 
	\begin{compactenum}
		\item[$\Rightarrow:$] $\{i,j\} \not\in \textbf{ Codomain}(\phi|_{A \cup B})$  means that none of the elements in $A$ or $B$ has the description of either $i$ or $j$. Thus the def.~\ref{def:desun2} yields that $A \mathop \bigcup \limits_{\phi=\{i,j\}} B=\emptyset$.
		\item[$\Leftarrow:$] As both the subsets $A,B \subset K$ are nonempty, $A \mathop \bigcup \limits_{\phi=\{i,j\}} B$ implies that none of the elements of $A$ or $B$ have the $i$ or $j$ as their description. Hence, $\{i,j\} \not\in \textbf{Codomain}(\phi|_{A \cup B})$.
	\end{compactenum}
Proving both $\Rightarrow$ and $\Leftarrow$ leads to the conclusion that $\Leftrightarrow$, or the bijection also holds.
	\item 
	\begin{compactenum}
		\item[$\Rightarrow:$] $i \not\in \textbf{ Codomain}(\phi|_{A \cup B})$ means that none of the elements in the sets $A$ or $B$ has the description $i$. This implies that all the elements of $A \cup B$ having the description of $i$ or $j$ are the same as those elements of $A \cup B$ having the description $j$.
		\item[$\Leftarrow:$] We know that $\phi$ is a function hence it can be either many-to-one or one-to-one, i.e. multiple elements can have the same description but it is not possible for an element to have multiple descriptions. $A \mathop \bigcup \limits_{\phi=\{i,j\}} B=A \mathop \bigcup \limits_{\phi=\{j\}} B$, i.e. elements of $A \cup B$ having the descriptions of $i$ or $j$ are same as the elements of $A \cup B$ having the description of $j$, if there are no elements having the description $i$. This means that $i \not\in \textbf{ Codomain}(\phi|_{A \cup B})$.
	\end{compactenum}
Proving both the $\Rightarrow$ and $\Leftarrow$ establishes that $\Leftrightarrow$ holds.
	\item 
		\begin{compactenum}
			\item[$\Rightarrow:$] $\phi$ being an injective function forces each element to have a unique description. Combining this with the fact that when restricted to $A \cap B,\, \phi=\{i,j\}$, it can be observed that it is not possible for any element outside of $A \cap B$ to be either $i$ or $j$. From def.~\ref{def:desun2}, this leads to $A \mathop \bigcup \limits_{\phi=\{i,j\}} B=A \cap B$.
			
			\item[$\Leftarrow:$] $A \mathop \bigcup \limits_{\phi=\{i,j\}} B=A \cap B$ requires that the elements in $A \cap B$ have description of either $i$ or $j$. This leads to the condition that $\phi|_{A \cap B}$. The additional requirement forced by the antecedant is that no element outside $A \cap B$ have the same description as that of elements inside. This calls for the map $\phi$ to preserve the distinctness of elements. Such a fucntion is called an injective function. 
		\end{compactenum}
	Hence, proving both $\Rightarrow$ and $\Leftarrow$ results in establishing the $\Leftrightarrow$.		
    \end{compactenum}
\end{proof}
\subsubsection{Restrictive and descriptive nondiscriminatory union}\label{sec:def3}
~\\
Another possible definition is to consider a dual of descriptive intersection defined in  Def.~\ref{def:desint} by replacing $\cup$ with $
\cap$. The resulting notion of union is restrictive, since it only considers the elements common to both sets $A$ and $B$. Moreover, we donot choose elements with specific descriptions, hence this notion is descriptive nondiscriminatory.
The resulting definition is as follows:
\begin{definition}\label{def:desun3}
	Let $A,B \subset K$ be subsets of $K$ and $\phi:2^K \rightarrow \mathbb{R}^n$. Then,
	\begin{align*}
	A \tilde{\mathop \bigcup \limits_\Phi} B=\{x \in A \cap B: \phi(s) \in \phi(A)\, or \, \phi(x) \in \phi(B)\},
	\end{align*}
	where $\tilde{\mathop \bigcup \limits_\Phi} $ is the spatially restricted and descriptively indiscriminant union.
\end{definition}
We can represent this definition as the following diagram.
\begin{align}
\xymatrix{
	A \ar[d]^\phi & A \cap B \ar[l]_a \ar[r]^b &B \ar[d]^\phi \\
	\phi(A) \ar@<1ex>[u]^\pi \ar[r]^c& \phi(A) \cup \phi(B) \ar[d]^\pi& \phi(B) \ar@<1ex>[u]^\pi \ar[l]_d \\
	& A \tilde{\mathop \bigcup \limits_\Phi} B & 
}
\end{align}
An important result regarding the spatially restricted and discriptively indiscriminant union is stated as follows.
\begin{theorem}\label{thm:equiv_inter}
	Let $A,B \subset K$ be the two sets in $K$ and $\phi:2^K \rightarrow \mathbb{R}^n$ be a probe function. Then,
	\begin{align*}
	  A \tilde{\mathop \bigcup \limits_\Phi} B \Leftrightarrow A \bigcap B .
	\end{align*}
\end{theorem}
\begin{proof}
	It can be seen that every $x \in A \cap B$ satisfies the condition that $\phi(x) \in \phi(A) \, and \, \phi(x) \in  \phi(B)$. Thus it  satisfies the condition that $\phi(x) \in \phi(A) \, or \, \phi(x) \in  \phi(B)$. Def.~\ref{def:desun3} reduces to $A \cap B$, hence proved.
\end{proof}
An other statement of the above result is that the description of every element in $A \cap B$ is either included in $\phi(A)$ or $\phi(B)$. Spatially restricted and descriptively indiscriminant union is equivalent to the usual interesction of sets ignoring the description of constituent elements.  

\subsubsection{Non-restrictive and descriptive nondiscriminatory union}\label{sec:def4}
~\\
We consider a definition of the descriptive union that considers all the unique elements of sets $A$ and $B$, i.e. $A \cup B$. Moreover, we do not filter the elements chosen based on their description and hence the notion of being descriptive nondiscriminatory. With this in mind we propose the following definition, 
\begin{definition}\label{def:desun4}
	Let $A,B \subset K$ be two subsets of $K$ and $\phi:2^K \rightarrow \mathbb{R}^n$ be the probe function. Then \begin{align*}
	A \mathop \bigcup \limits_\Phi B=\{x\in A \cup B: \phi(x)\in \phi(A) \, or \, \phi(x) \in \phi(B) \}
	\end{align*}
	where $\mathop \bigcup \limits_\Phi$ is the spatially unrestricted and descriptively indiscriminant union.
\end{definition}
This can be illustrated using the following diagram.
\begin{align}
\xymatrix{
	A \ar[d]^\phi & A \cup B \ar[l]_a \ar[r]^b &B \ar[d]^\phi \\
	\phi(A) \ar@<1ex>[u]^\pi \ar[r]^c& \phi(A) \cup \phi(B) \ar[d]^\pi& \phi(B) \ar@<1ex>[u]^\pi \ar[l]_d \\
	& A \mathop \bigcup \limits_\Phi B & 
}
\end{align}
An important result regarding the spatially unrestricted and descriptively indiscriminant union is presented in the following theorem.
\begin{theorem}\label{thm:equiv_union}
	Let $A,B \subset K$ be the two sets in $K$ and $\phi:2^K \rightarrow \mathbb{R}^n$ be a probe function. Then,
	\begin{align*}
	A \mathop \bigcup \limits_\Phi B \Leftrightarrow A \bigcup B.
	\end{align*}
\end{theorem} 
\begin{proof}
	It can be seen that every $x \in A \cup B$ satisfies the condition that $\phi(x) \in \phi(A) \, or \, \phi(x) \in  \phi(B)$. Thus, def.~\ref{def:desun4} reduces to $A \cup B$, hence proved.
\end{proof}
We can restate this principle by saying that description of elements in $A \cup B$ are the same as $\phi(A)$ or $\phi(B)$.  Thus, a spatially unrestricted and descriptvely indiscriminant union is equivalent to the usual notion of union, that ignores the descriptions of elements.

\section{Applications}
In this section we will use the notion of descriptive union developed in section~\ref{sec:desunion}. We will try to emulate some concepts developed in literature using the set theoretic union. This is to start a study in to wealth of structures that descriptive unions have to offer. 
\subsection{Convex union representable complexes}
A collection of sets $X$, such that for a set $S \in X$ any subset $T \subset S$ is also in $X$, is called a \emph{simplicial complex}(here after referred to as a complex).
A \emph{nerve} complex is a collection of sets with a non-empty intersection.  The notion of a nerve complex was introduced by P. Alexandroff~\cite{Alexandroff1926MAnnNerfTheorem},\cite{Alexandroff1932elementaryConcepts}, elaborated by P. Alexandroff and H. Hopf~\cite{AlexandroffHopf1935Topologie} in the context of CW complexes and generalized by H. Edelsbrunner and J.L. Harer~\cite{Edelsbrunner1999}.     If a complex can be represented as a nerve of convex sets in $\mathbb{R}^d$, then it is \emph{d-representable}. An extension of this concept is defined in \cite{jeffs2018convex}. A \emph{d-convex union representable} complex is a nerve of finite collection of convex open sets in $\mathbb{R}^d$, with the additional property that the union of these sets is also convex.

Let us define this idea in the context of descriptive set theoretic operations. We begin by defining the idea of a descriptive nerve as follows:
\begin{definition}\label{def:nerve_des}
	Let $\mathcal{U}$ be a finite collection of sets, $\{U_1,\cdots,U_n\}$. Then, the descriptive nerve is defined as:
	\begin{align*}
	\mathop \mathcal{N} \limits_{\Phi}(\mathcal{U})=\{\{U_i\}_{i \subset \{1,\cdots,n\}}: \mathop \bigcap \limits_{\Phi} U_i \neq \emptyset \}
	\end{align*}
\end{definition}  
As we have four different types of descriptive unions we have four types of convex union representable complexes. Here, we use the symbol $\mathop \bigcup \limits_{des}$ to represent all of these in the definition of descriptive d-convex union representable complexes.
\begin{definition}\label{def:dconvexun_des}
	Let $K$ be a simplicial complex and $\mathcal\{U\}$ be a collection of convex open sets $\{U_1,\cdots,U_n\} \subset \mathbb{R}^d$. If $K=\mathop \mathcal{N} \limits_{\Phi}(\mathcal{U})$ and $\mathop \bigcup \limits_{des} U_i$ is also convex for all $i=1,\cdots,n$. Then, $K$ is a d-convex union representable complex.
\end{definition}  
In the following theorem we formulate some results regarding convexity of the descriptive union of convex sets. We formulate this theorem for the unions of two sets as in the current paper our focus is on union as a binary operation.
\begin{theorem}\label{thm:desun_convx}
	Let $A,B \subset K$ be two convex sets, let $\phi:2^K \rightarrow \mathbb{R}^n$ be a probe function and $\pi:\mathbb{R}^n \rightarrow 2^K$ be a map such that $\pi: x \mapsto \{y \in K: \phi(y)=x\}$. Then the following properties are true regarding the convexity of descriptive union:
	\begin{compactenum}[$1^o$]
		\item $\tilde{\pi}(i) \cup \tilde{\pi}(j) \text{ is convex} 
		\Leftrightarrow A \tilde{\mathop \bigcup \limits_{\phi=\{i,j\}}} B $ is convex, where  $\tilde{\pi}=\pi \rhd \{2^{A \cap B}\}$.
		\item $\hat{\pi}(i)  \cup \hat{\pi}(j) \text{ is convex} 
		\Leftrightarrow A \mathop \bigcup \limits_{\phi=\{i,j\}} B $ is convex, where $\hat{\pi}=\pi \rhd \{2^{A \cup B}\}$
		\item $A \tilde{\mathop \bigcup \limits_\Phi} B$ is convex
		\item $A \cup B \text{ is convex} \Leftrightarrow A \mathop \bigcup \limits_\Phi B$ is convex	
	\end{compactenum}
where $f \rhd D$ is restriction of the range of a function $f$ to $D$, a subset of the original range of the function.
\end{theorem} 
\begin{proof}
	\begin{compactenum}[$1^o$]
		\item  From Def.\ref{def:desun1} it can be concluded that $A \tilde{\mathop \bigcup \limits_{\phi=\{i,j\}}} B$ is equivalent to $x \in A \cap B$ such that $\phi(x)=\{i,j\}$, that is equivalent to $\tilde{\pi(i)} \cup \tilde{\pi(j)}$. $\tilde{\pi}=\pi \rhd A \cap B$, is the range restriction of $\pi$ to $2^{A \cap B}$. This is required as the range of $\pi$ is $2^K$. Once, we establish this equivalence we can conclude that if  $\tilde{\pi(i)} \cup \tilde{\pi(j)}$ is convex then $A \tilde{\mathop \bigcup \limits_{\phi=\{i,j\}}} B$ is also convex and vice aversa. Hence, we can prove the bijection. 
		
		\item It can be established from Def.~\ref{def:desun2} that $A \mathop \bigcup \limits_{\phi=\{i,j\}} B$ is equivalent to $x \in A \cup B$ such that $\phi(x) \in \{i,j\}$. This can be expressed as $\hat{\pi(i)} \cup \hat{\pi(j)}$, $\hat{\pi}=\pi \rhd 2^{A \cup B}$ is the restriction of range of $\pi$ to $2^{A \cup B}$. This is necessary as the range of $\pi$ is $2^K$. Once this equivalence is established, we can immediately conclude that if $\hat{\pi(i)} \cup \hat{\pi(j)}$ is convex then $A \mathop \bigcup \limits_{\phi=\{i,j\}} B$ and vice aversa, thus proving the bijection.     
		
		\item It can be established from Thm.~\ref{thm:equiv_inter} that $A \tilde{\mathop \bigcup \limits_\Phi} B$ is equivalent to $A \cap B$. The intersection of convex sets is convex due to a well established result. If the intersection contains a single point it is true by definition. In case of more than one points, consider any two $x$ and $y$. The definition of intersection constrains the points to lie in both $A$ and $B$. As both $A$ and $B$ are convex thus a line joining the points would lie in both the sets and hence in $A \cap B$. Thus, the intersection of convex sets is also convex. 
		
		\item It can be established from Thm.~\ref{thm:equiv_union} that $A \mathop \bigcup \limits_\Phi B$ is equivalent to $A \cup B$. Thus, if $A \cup B$ is convex this implies that $A \mathop \bigcup \limits_\Phi $ is convex and vice versa. Hence, proving the bijection. 
		
		\end{compactenum}
\end{proof}
\subsection{Descriptive unions in digital images }
In this section we consider the application of different descriptive unions to digital images. As mentioned earlier in section \ref{sec:desint}, descriptive set theory started with the notion of an intersection, motivated by the fact that sets in digital images can have similar color eventhough they may be spatially far. Following is a brief account regarding each of the different unions defined in section \ref{sec:desunion}. We use one of the stock images used in MATLAB, called \emph{peppers.png}, showin in Fig.~\ref{subfig:orig}. This image has objects with different intensities, which for the purpose of this study we take to be the description. Let us write this down in mathematical notation, consistent with the one used in this study.

Let $\mathcal{I}$ be the spatial region over which the image is defined, then $2^\mathcal{I}$ be the set of all the possible regions in it. This spatial region is usually a subset of the Euclidean plane $\mathbb{R}^2$, but different regions can be used if the application dictates so. To put this in the context of digital images, the smallest region in an image that can have a distinct description is a \emph{pixel}. In digital photography a sensor array measures the intensity of light falling on it via the lens. Each sensor measures three values namely \emph{red,blue and green} light intensities for each pixel. This act of capturing an image can be represented as a map, $\phi:2^\mathcal{I} \rightarrow \mathbb{R}^3$. Thus, a digital image is a \emph{glossa} as defined in section \ref{sec:intro}, and can be written as $\mathcal{I}_\Phi$. A fibre bundle structure, $(\mathcal{I}_\Phi,\mathcal{I},\pi,\phi(U))$, can be used to specify the relationship between image, spatial region over which it is defined and the sensors used to acquire it. Here, $\mathcal{I}_\Phi$ is the digital image, $\mathcal{I}$ is the region over which the image is defined, $U \subset \mathcal{I}$, $\phi:2^\mathcal{I} \rightarrow \mathbb{R}^3$ models the imaging sensors and $\pi:\mathcal{I}_\Phi \rightarrow \mathcal{I}$ is a continuous surjection. This structure satisfies Eqns.~\ref{eq:fibre} and \ref{eq:localtriv}.

Moreover, in digital images there are a lot of artifacts that effect intensity values captured by the sensors. These include but are not limited to quantization noise, uneven illumination and electronic noise etc. To take these into account we use a \emph{tolerance} value when equating intensity values. Instead of $\|\phi(A)-\phi(B)\|_2=0 \Rightarrow \phi(A) \simeq \phi(B)$, we consider $\|\phi(A)-\phi(B)\|_2 \leq \eta \Rightarrow \phi(A) \simeq \phi(B)$. Here, $\eta$ is an arbitrary positive number and $A,B \subset \mathcal{I}$.

\begin{figure}
	\centering
	\begin{subfigure}[Original image]
		{\includegraphics[width=2in]{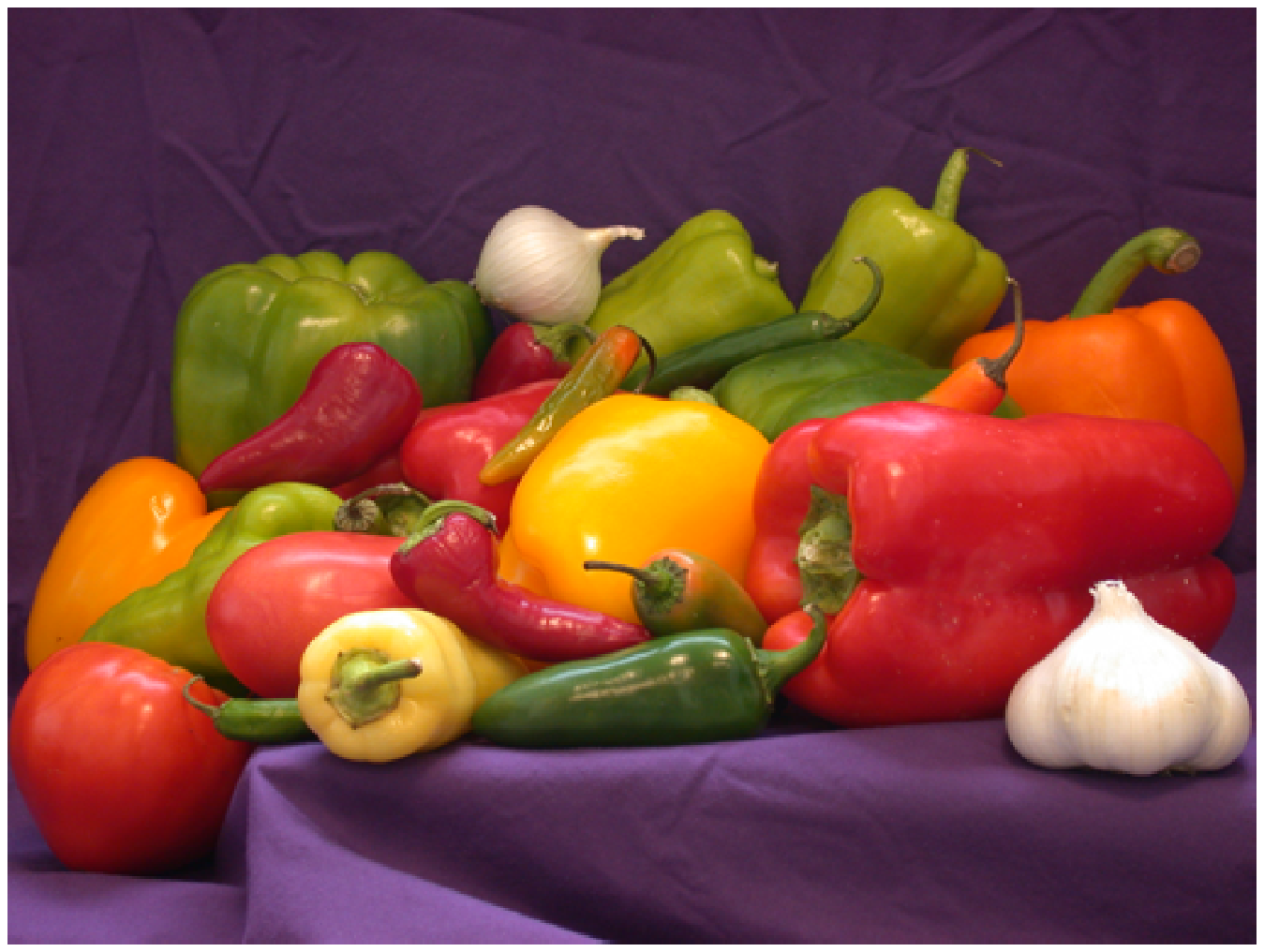}
			\label{subfig:orig}}
	\end{subfigure}
	\begin{subfigure}[Two subsets in the image]
		{\includegraphics[width=2in]{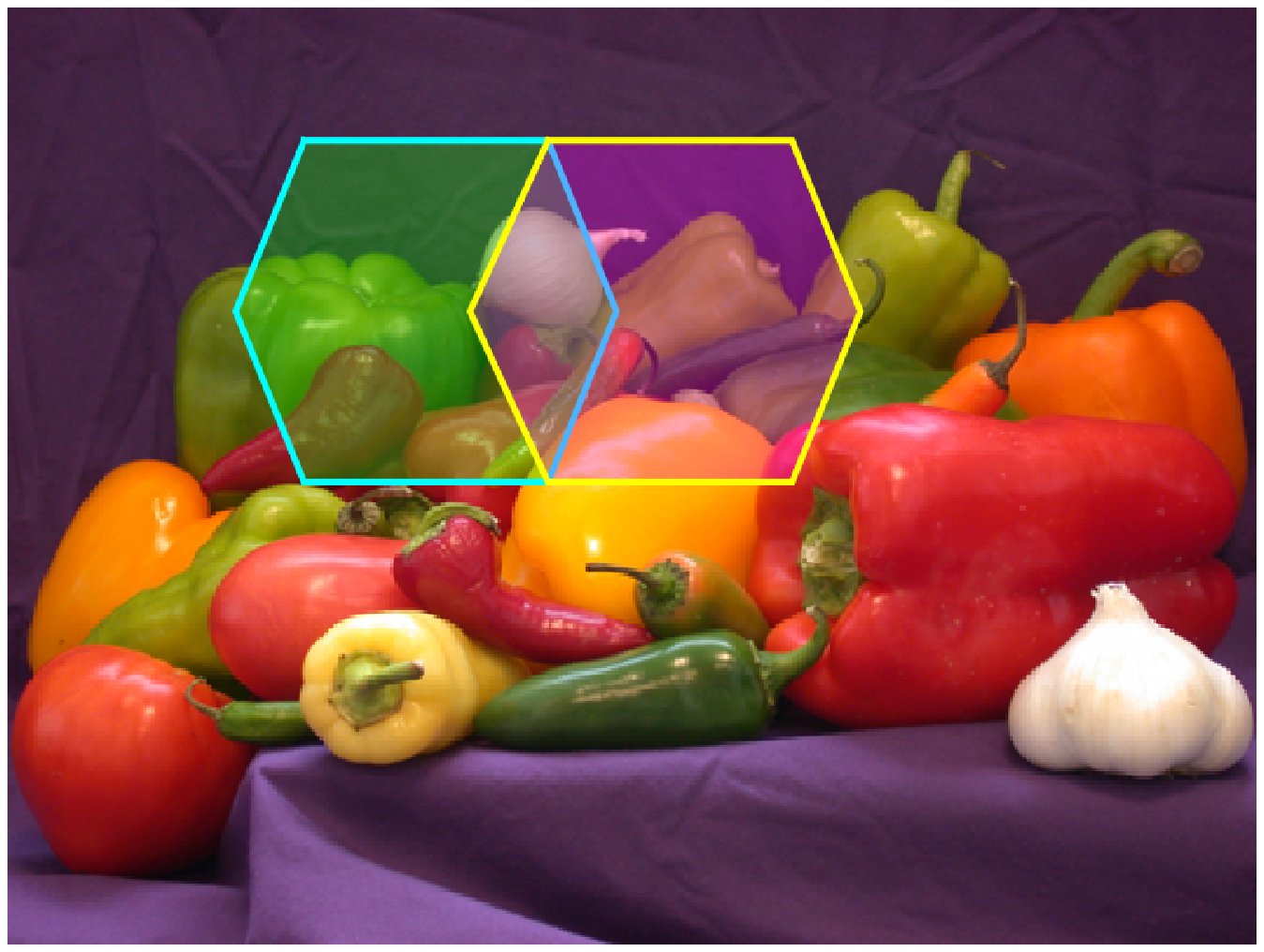}
			\label{subfig:neigh}}
	\end{subfigure}
	\begin{subfigure}[$\tilde{\mathop \bigcup \limits_{\eta,\phi=\{[254, 224, 198],[208, 35, 37]\}}}$]
		{\includegraphics[width=2in]{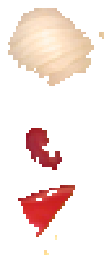}
			\label{subfig:desun1}}
	\end{subfigure}
	\begin{subfigure}[$\mathop \bigcup \limits_{\eta,\phi=\{[254, 224, 198],[208, 35, 37]\}}$]
		{\includegraphics[width=2in]{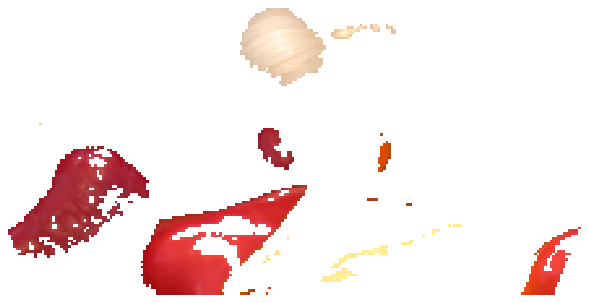}
			\label{subfig:desun2}}
	\end{subfigure}
	\begin{subfigure}[$\mathop \bigcup \limits_{\eta,\Phi}$]
		{\includegraphics[width=2in]{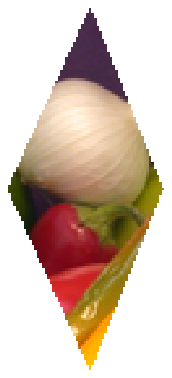}
			\label{subfig:desun3}}
	\end{subfigure}
	\begin{subfigure}[$\tilde{\mathop \bigcup \limits_{\eta,\Phi}}$]
		{\includegraphics[width=2in]{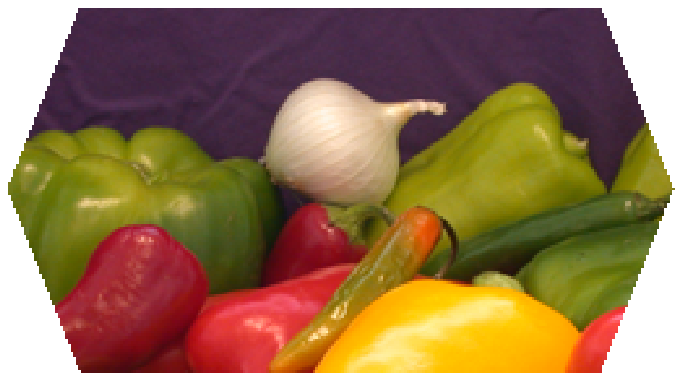}
			\label{subfig:desun4}}
	\end{subfigure}
	
	\caption{This figure illustrates how the different descriptive unions defined in section \ref{sec:desunion}, work in the case of a digital image presented in Fig.~\ref{subfig:orig}. The hexagonal  intersecting subsets, that we select are shown in Fig.~\ref{subfig:neigh}. These are colored green and red for the purpose of identification. The values of parameters for this experiment are $\eta=60$ and $\phi=\{[254,224,198],[208,35,37]\}$. Fig.~\ref{subfig:desun1} is the illustration of restrictive and descriptive discriminatory union. Fig.~\ref{subfig:desun2} shows the non-restrictive and descriptive discriminatory union. Fig.~\ref{subfig:desun3} is the restrictve and descriptive non-discriminatory union. Fig.~\ref{subfig:desun4} is the non-restrictive and descriptive non-discriminatory union.}
\end{figure}

\subsubsection{Restrictive and descriptive discriminatory union} 
~\\
Let us present a modified version of the original definition(def.~\ref{def:desun1}) to incorporate tolerance.
\begin{definition}\label{def:desun1tol}
	Let $A \subset K$ be a subset in $K$,$\phi:2^K \rightarrow \mathbb{R}^n$ be a probe function and $\eta$ be an arbitrary positive number. Then
	\begin{align*}
	A \tilde{\mathop \bigcup \limits_{\eta,\phi=\{i,j\}}} B = \{x \in A \cap B: \|\phi(x)-i\|_2 \leq \eta \, or \,\|\phi(x)-j\|_2 \leq \eta \}, 
	\end{align*}
	where $\tilde{\mathop \bigcup \limits_{\eta, \phi=\{i,j\}}}$ is the spatially restricted and descriptively discriminant union with a tolerance.
\end{definition}

Now we move on to the application of this concept to the digital image shown in Fig.~\ref{subfig:orig}. We select two subsets with a non-empty intersection as shown in Fig.~\ref{subfig:neigh}, one of the sets is coloured green and the other one red for identification. From def.~\ref{def:desun1tol}, it can be established that we are restricted to the intersection (a diamond shaped region) of regions under consideration. For this experiment we set $\eta= 60$ and $\phi=\{[254, 224, 198],[208, 35, 37]\}$. Thus, $A \tilde{\mathop \bigcup \limits_{\eta,\phi=\{i,j\}}} B$ is the set of pixels in $A \cup B$ within the tolerance range of $[254, 224, 198]$ or $[208, 35, 37]$. This is shown in Fig.~\ref{subfig:desun1}.

\subsubsection{Non-restrictive and descriptive discriminatory union}
~\\
We start by presenting a version of def.~\ref{def:desun2} that incorporates tolerance.
\begin{definition}\label{def:desun2tol}
	Let $A,B \subset K$ be subsets in $K$, $\phi:2^K \rightarrow \mathbb{R}^n$ be a probe function and $\eta$ be an arbitrary positive number. Then
	\begin{align*}
	A \mathop \bigcup \limits_{\eta,\phi=\{i,j\}} B = \{x \in A \cup B: \|\phi(x)-i\|_2 \leq \eta \, or \,\|\phi(x)-j\|_2 \leq \eta\}, 
	\end{align*}
	where $\mathop \bigcup \limits_{\eta,\phi=\{i,j\}}$ is the spatially unrestricted and descriptively discriminant union with a tolerance.
\end{definition}
Similar to the previous experiment we set $\eta=60$ and $\phi=\{[254,224,198],[208,35,37]\}$. We select two subsets $A,B$ as shown in Fig.~\ref{subfig:neigh}. Def.~\ref{def:desun2tol} states that $A \mathop \bigcup \limits_{\eta,\phi=\{i,j\}} B$ lies within $A \cup B$ and consists of the pixels with intensity in the tolerance range of either $[254,224,198]$ or $[208,35,37]$. We can see this illustrated in Fig.~\ref{subfig:desun2}.

\subsubsection{Restrictive and descriptive nondiscriminatory union}
~\\
A version of def.~\ref{def:desun3} adjusted to account for tolerance is as follows.
\begin{definition}\label{def:desun3tol}
	Let $A,B \subset K$ be subsets of $K$, $\phi:2^K \rightarrow \mathbb{R}^n$ and $\eta$ be an arbitrary positive number. Then,
	\begin{align*}
	A \tilde{\mathop \bigcup \limits_{\eta,\Phi}} B=\{x \in A \cap B: \|\phi(s)-\phi(A)\|_2 \leq \eta\, or \, \|\phi(s)-\phi(B)\|_2 \leq \eta\},
	\end{align*}
	where $\tilde{\mathop \bigcup \limits_{\eta,\Phi}} $ is the spatially restricted and descriptively indiscriminant union with a tolerance.
\end{definition} 
The two subsets under consderation are presented in Fig.~\ref{subfig:neigh}. It can be seen that $A \tilde{\mathop \bigcup \limits_{\eta,\Phi}} B$ is restricted to $A \cap B$. It consists of all the pixels in $A \cap B$ that have the description matching any of the pixels in $A$ or $B$. A result stating the equivelence of this union to the $A \cap B$  was theorized for the case without tolerance in Thm.~\ref{thm:equiv_inter}. We can reformulate this for the case in which we consider tolerance as follows.
\begin{theorem}\label{thm:equiv_intertol}
	Let $A,B \subset K$ be the two sets in $K$, $\phi:2^K \rightarrow \mathbb{R}^n$ be a probe function and $\eta$ be an arbitrary positive number. Then,
	\begin{align*}
	A \tilde{\mathop \bigcup \limits_{\eta,\Phi}} B \Leftrightarrow A \bigcap B .
	\end{align*}
\end{theorem}
\begin{proof}
	It can be seen that every $x \in A \cap B$ satisfies the condition that $\phi(x) \in \phi(A) \, and \, \phi(x) \in  \phi(B)$. These can be rewritten as $\|\phi(x)-\phi(A)\|_2 = 0$ and $\|\phi(x)-\phi(B)\|_2 =0$. These conditions are more restrictive than the ones required for $A \tilde{\mathop \bigcup \limits_{\eta,\Phi}} B$ namely, $\|\phi(x)-\phi(A)\|_2 \leq eta$ and $\|\phi(x)-\phi(B)\|_2  \leq \eta$, as $\eta>0$. Thus, any element $x$ satisfying $\|\phi(x)-\phi(A)\|_2 = 0$ and $\|\phi(x)-\phi(B)\|_2 =0$, satifies $\|\phi(x)-\phi(A)\|_2 \leq \eta$ and $\|\phi(x)-\phi(B)\|_2  \leq \eta$. Thus, Def.~\ref{def:desun3tol} reduces to $A \cap B$, hence proved.
\end{proof}
We can see the theorem in action when studying this particular union for digital images. The parameters $\eta=60$ and $\phi={[254,224,198],[208,35,37]}$ are set as in the previous cases. The sets selected are shown in Fig.~\ref{subfig:neigh}. We can see from Thm.~\ref{thm:equiv_intertol} that $A \tilde{\mathop \bigcup \limits_{\eta,\Phi}} B$ is the same as $A \cap B$, as shown in Fig.~\ref{subfig:desun3}. 
\subsubsection{Non-restrictive and descriptive nondiscriminatory union}
~\\
Modification of def.~\ref{def:desun4} to include tolerance is given below.
\begin{definition}\label{def:desun4tol}
	Let $A,B \subset K$ be two subsets of $K$, $\phi:2^K \rightarrow \mathbb{R}^n$ be the probe function and $\eta$ be an arbitrary positive number. Then 
	\begin{align*}
	A \mathop \bigcup \limits_{\eta,\Phi} B=\{x\in A \cup B: \|\phi(s)-\phi(A)\|_2 \leq \eta\, or \, \|\phi(s)-\phi(B)\|_2 \leq \eta \}
	\end{align*}
	where $\mathop \bigcup \limits_{\eta,\Phi}$ is the spatially unrestricted and descriptively indiscriminant union with a tolerance.
\end{definition}
 The parameter values are set for this experiment to $\eta=60$ and $\phi=\{[254,224,198], \linebreak[0] [208,35,37]\}$ and the subsets used are the same as those shown in Fig.~\ref{subfig:neigh}. We can see this in def.~\ref{def:desun4tol}, that $A \mathop \bigcup \limits_{\eta,\Phi} B$ are the pixels inside of $A \cup B$ with intensities in the tolerance range of $[254,224,198]$ or $[208,35,37]$. Thm.~\ref{thm:equiv_union} states that this union is equivalent to $A \cup B$ for the case without tolerance. We can formulate a similar result for the case with tolerance.
 \begin{theorem}\label{thm:equiv_uniontol}
 	Let $A,B \subset K$ be the two sets in $K$, $\phi:2^K \rightarrow \mathbb{R}^n$ be a probe function, and $\eta$ is an arbitrary positive number. Then,
 	\begin{align*}
 	A \mathop \bigcup \limits_{\eta, \Phi} B \Leftrightarrow A \bigcup B.
 	\end{align*}
 \end{theorem} 
 \begin{proof}
 	It can be seen that every $x \in A \cup B$ satisfies the condition that $\phi(x) \in \phi(A) \, or \, \phi(x) \in  \phi(B)$. These statements are equivalent to $\|\phi(x)-\phi(A)\|_2 =0 \, or \, \|\phi(x)-\phi(B)\|_2=0$. The conditions for $x\ in A \cup B$ to be in $A \mathop \bigcup \limits_{\eta, \Phi} B$ are $\|\phi(x)-\phi(A)\|_2 \leq \eta \, or \, \|\phi(x)-\phi(B)\|_2 \leq \eta$. It is straigh forward that $\|\phi(x)-\phi(A)\|_2 =0 \Rightarrow \|\phi(x)-\phi(A)\|_2 \leq \eta $ and $\|\phi(x)-\phi(B)\|_2 =0 \Rightarrow \|\phi(x)-\phi(B)\|_2 \leq \eta $. Thus, def.~\ref{def:desun4tol} reduces to $A \cup B$, hence proved.
 \end{proof}
 We can see this equivalence illustrated for the case of digital images in Fig.~\ref{subfig:desun4}
 
 \section{Conclusion}
This study extends the notion of descriptive intersection and presents it in the newly defined framework of descriptive unions called glossa.  We use fibre bundles to explain the relationship between a set and the description of its elements. We propose four different notions of descriptive union and determine their properties. Moreover, we consider the application of these unions in \emph{d-convex union representable complexes} and in digital images.  
\bibliographystyle{amsplain}
\bibliography{desrefs}

\providecommand{\bysame}{\leavevmode\hbox to3em{\hrulefill}\thinspace}
\providecommand{\MR}{\relax\ifhmode\unskip\space\fi MR }
\providecommand{\MRhref}[2]{%
  \href{http://www.ams.org/mathscinet-getitem?mr=#1}{#2}
}
\providecommand{\href}[2]{#2}
\begin{thebibliography}{10}

\bibitem{AhmadPeters2017TAMCSspokes}
M.Z. Ahmad and J.F. Peters, \emph{Proximal cech complexes in approximating
  digital image object shapes. theory and application}, Theory Appl. Math.
  Comput. Sci. \textbf{7} (2017), no.~2, 81--123, MR3769444.

\bibitem{Alexandroff1926MAnnNerfTheorem}
P.~Alexandroff, \emph{Simpliziale approximationen in der allgemeinen
  topologie}, Mathematische Annalen \textbf{101} (1926), no.~1, 452--456,
  MR1512546.

\bibitem{Alexandroff1932elementaryConcepts}
\bysame, \emph{Elementary concepts of topology}, Dover Publications, Inc., New
  York, 1965, 63 pp., translation of Einfachste Grundbegriffe der Topologie
  [Springer, Berlin, 1932], translated by Alan E. Farley , Preface by D.
  Hilbert, MR0149463.

\bibitem{AlexandroffHopf1935Topologie}
P.~Alexandroff and H.~Hopf, \emph{Topologie. {B}and i}, Springer, Berlin, 1935,
  Zbl 13, 79; reprinted Chelsea Publishing Co., Bronx, N. Y., 1972. iii+637
  pp., MR0345087.

\bibitem{DiConcilio2018MCSdescriptiveProximities}
A.~Di Concilio, C.~Guadagni, J.F. Peters, and S.~Ramanna, \emph{Descriptive
  proximities. properties and interplay between classical proximities and
  overlap}, Math. Comput. Sci. \textbf{12} (2018), no.~1, 91--106, MR3767897.

\bibitem{Edelsbrunner1999}
H.~Edelsbrunner and J.L. Harer, \emph{Computational topology. an introduction},
  Amer. Math. Soc., Providence, RI, 2010, {x}ii+241 pp. ISBN:
  978-0-8218-4925-5, MR2572029.

\bibitem{jeffs2018convex}
R~Amzi Jeffs and Isabella Novik, \emph{Convex union representability and convex
  codes}, arXiv preprint arXiv:1808.03992 (2018).

\bibitem{Peters2007AMS}
J.F. Peters, \emph{Near sets. {G}eneral theory about nearness of sets}, Applied
  Math. Sci. \textbf{1} (2007), no.~53, 2609--2629.

\bibitem{Peters2007FINearSets}
\bysame, \emph{Near sets. {S}pecial theory about nearness of objects},
  Fundamenta Informaticae \textbf{75} (2007), 407--433, MR2293708.

\bibitem{peters2013mcs}
\bysame, \emph{Local near sets: {P}attern discovery in proximity spaces}, Math.
  in Comp. Sci. \textbf{7} (2013), no.~1, 87--106, DOI
  10.1007/s11786-013-0143-z, MR3043920.

\bibitem{Peters2016ComputationalProximity}
\bysame, \emph{Computational proximity. {E}xcursions in the topology of digital
  images.}, Intelligent Systems Reference Library \textbf{102} (2016), xxviii +
  433pp, DOI: 10.1007/978-3-319-30262-1,MR3727129 and Zbl 1382.68008.

\bibitem{Peters2018AlMSproximalPlanarShapes}
\bysame, \emph{Proximal planar shapes. correspondence between triangulated
  shapes and nerve complexes}, Bulletin of the Allahabad Mathematical Society
  \textbf{33} (2018), 113--137, MR3793556.

\bibitem{Peters2018JMSMvortexNerves}
\bysame, \emph{Proximal vortex cycles and vortex nerve structures.
  non-concentric, nesting, possibly overlapping homology cell complexes}, J.
  Math. Sci. and Applications (2018), 1--17, \emph{in press}, See, also,
  \url{https://arxiv.org/abs/1805.03998}.

\bibitem{Peters2016MNC}
J.F. Peters and E.~\.{I}nan, \emph{Strongly proximal edelsbrunner-harer
  nerves}, Proceedings of the Jangjeon Mathematical Society \textbf{19} (2016),
  no.~3, 563--582, MR3618825.

\bibitem{PetersTozzi2916quantumEntanglement}
J.F. Peters and A.~Tozzi, \emph{Quantum entanglement on a hypersphere},
  Internat. J. Theoret. Phys. \textbf{55} (2016), no.~8, 3689--3696, MR3518899,
  Zbl 1361.81025.

\end{thebibliography}
\end{document}